\documentclass[11pt,a4paper]{article}%
\usepackage{amsfonts}
\usepackage{color}
\usepackage{amsmath}
\usepackage{fullpage}
\usepackage{amssymb}
\usepackage{revsymb4-1}
\usepackage{graphicx}
\usepackage{hyperref}%
\providecommand{\U}[1]{\protect\rule{.1in}{.1in}}

\newtheorem{theorem}{Theorem}

\newtheorem{corollary}[theorem]{Corollary}

\newtheorem{definition}[theorem]{Definition}
\newtheorem{example}[theorem]{Example}

\newtheorem{lemma}[theorem]{Lemma}

\newtheorem{proposition}[theorem]{Proposition}
\newtheorem{remark}[theorem]{Remark}

\def\squareforqed{\hbox{\rlap{$\sqcap$}$\sqcup$}}
\def\qed{\ifmmode\squareforqed\else{\unskip\nobreak\hfil
\penalty50\hskip1em\null\nobreak\hfil\squareforqed
\parfillskip=0pt\finalhyphendemerits=0\endgraf}\fi}
\def\endenv{\ifmmode\;\else{\unskip\nobreak\hfil
\penalty50\hskip1em\null\nobreak\hfil\;
\parfillskip=0pt\finalhyphendemerits=0\endgraf}\fi}
\newenvironment{proof}{\noindent \textbf{{Proof~} }}{\qed}
\mathchardef\ordinarycolon\mathcode`\:
\mathcode`\:=\string"8000
\def\vcentcolon{\mathrel{\mathop\ordinarycolon}}
\begingroup \catcode`\:=\active
  \lowercase{\endgroup
  \let :\vcentcolon
  }

\newcommand{\nc}{\newcommand}
\nc{\rnc}{\renewcommand}
\nc{\beq}{\begin{equation}}
\nc{\eeq}{\end{equation}}
\nc{\beqa}{\begin{eqnarray}}
\nc{\eeqa}{\end{eqnarray}}
\nc{\lbar}[1]{\overline{#1}}
\nc{\bra}[1]{\langle#1|}
\nc{\ket}[1]{|#1\rangle}
\nc{\ketbra}[2]{|#1\rangle\!\langle#2|}
\nc{\braket}[2]{\langle#1|#2\rangle}
\nc{\proj}[1]{| #1\rangle\!\langle #1 |}
\nc{\avg}[1]{\langle#1\rangle}
\nc{\Rank}{\operatorname{Rank}}
\nc{\smfrac}[2]{\mbox{$\frac{#1}{#2}$}}
\nc{\tr}{\operatorname{Tr}}
\nc{\ox}{\otimes}
\nc{\dg}{\dagger}
\nc{\dn}{\downarrow}
\nc{\cA}{{\cal A}}
\nc{\cB}{{\cal B}}
\nc{\cC}{{\cal C}}
\nc{\cD}{{\cal D}}
\nc{\cE}{{\cal E}}
\nc{\cF}{{\cal F}}
\nc{\cG}{{\cal G}}
\nc{\cH}{{\cal H}}
\nc{\cI}{{\cal I}}
\nc{\cJ}{{\cal J}}
\nc{\cK}{{\cal K}}
\nc{\cL}{{\cal L}}
\nc{\cM}{{\cal M}}
\nc{\cN}{{\cal N}}
\nc{\cO}{{\cal O}}
\nc{\cP}{{\cal P}}
\nc{\cR}{{\cal R}}
\nc{\cS}{{\cal S}}
\nc{\cT}{{\cal T}}
\nc{\cX}{{\cal X}}
\nc{\cY}{{\cal Y}}
\nc{\cZ}{{\cal Z}}
\nc{\csupp}{{\operatorname{csupp}}}
\nc{\qsupp}{{\operatorname{qsupp}}}
\nc{\var}{{\operatorname{var}}}
\nc{\rar}{\rightarrow}
\nc{\lrar}{\longrightarrow}
\nc{\polylog}{{\operatorname{polylog}}}
\nc{\1}{{\openone}}
\nc{\wt}{{\operatorname{wt}}}
\nc{\av}[1]{{\left\langle {#1} \right\rangle}}

\def\ll{\lambda}

\def\U{\Upsilon}

\nc{\RR}{{{\mathbb R}}}
\nc{\CC}{{{\mathbb C}}}
\nc{\FF}{{{\mathbb F}}}
\nc{\NN}{{{\mathbb N}}}
\nc{\ZZ}{{{\mathbb Z}}}
\nc{\PP}{{{\mathbb P}}}
\nc{\QQ}{{{\mathbb Q}}}
\nc{\UU}{{{\mathbb U}}}
\nc{\EE}{{{\mathbb E}}}
\nc{\id}{{\operatorname{id}}}

\nc{\CHSH}{{\operatorname{CHSH}}}

\nc{\be}{\begin{equation}}
\nc{\ee}{{\end{equation}}}
\nc{\bea}{\begin{eqnarray}}
\nc{\eea}{\end{eqnarray}}
\nc{\<}{\langle}
\rnc{\>}{\rangle}
\nc{\Hom}[2]{\mbox{Hom}(\CC^{#1},\CC^{#2})}
\nc{\rU}{\mbox{U}}

\nc{\ob}[1]{#1}

\def\id{\operatorname{id}}

\begin{document}

\title{\textbf{Weak locking capacity of quantum channels \protect\\ can be much larger than private capacity}}
\author{Andreas Winter\\
\textit{{\small {ICREA \& F\'{\i}sica Te\`{o}rica: Informaci\'{o} i Fenomens Qu\`{a}ntics}}}\\
\textit{{\small {Universitat Aut\`{o}noma de Barcelona,}}} 
\textit{{\small {ES-08193 Bellaterra (Barcelona), Spain}}}\\
\texttt{{\small {<andreas.winter@uab.cat>}}}
}

\date{20 July 2015}

\maketitle

\begin{abstract}
  We show that it is possible for the so-called \emph{weak locking capacity} 
  of a quantum channel 
  [Guha \emph{et al.}, \emph{Phys.~Rev.~X} {\bf 4}:011016, 2014] 
  to be much larger than its private capacity. Both reflect different ways
  of capturing the notion of reliable communication via a quantum
  system while leaking almost no information to an eavesdropper; the
  difference is that the latter imposes an intrinsically quantum security
  criterion whereas the former requires only a weaker, classical 
  condition.
  The channels for which this separation is most straightforward to
  establish are the complementary channels of classical-quantum
  (cq-)channels, and hence a subclass of Hadamard channels.
  We also prove that certain symmetric channels (related to photon number
  splitting) have positive weak locking capacity in the presence of
  a vanishingly small pre-shared secret, whereas their private capacity is zero.
  
  These findings are powerful illustrations of the difference between
  two apparently natural notions of privacy in quantum systems, relevant
  also to quantum key distribution (QKD): 
  the older, na\"{i}ve one based on accessible information, contrasting
  with the new, composable one embracing the quantum nature of the 
  eavesdropper's information.
  
  Assuming an additivity conjecture for constrained minimum output
  R\'{e}nyi entropies, the techniques of the first part demonstrate
  a single-letter formula for the weak locking capacity of complements
  to cq-channels, coinciding with a general upper bound of Guha \emph{et al.}
  for these channels.
  Furthermore, still assuming this additivity conjecture, this upper bound 
  is given an operational interpretation for general channels as the maximum 
  weak locking capacity of the channel activated by a suitable noiseless channel.
\end{abstract}

\bigskip

%\begin{flushright}
%  \emph{Videmus nunc per speculum in aenigmate}%: tunc autem...}
%\end{flushright}

\section{Introduction}
\label{sec:intro}
Information locking~\cite{I-locking} remains one of the most 
curious manifestations of the quantum nature of information,
which is in contrast to our (human) exclusively classical 
access to it. It is based on the simple (yet nontrivial) 
observation that the accessible information in a non-orthogonal
ensemble through a measurement can be smaller, indeed much smaller,
than the Holevo information. This occurs already for two
mutually unbiased bases, by the Maassen-Uffink entropic uncertainty
relation~\cite{MaassenUffink}, and the crucial realization is that
availability of the basis information (one bit) before the
measurement is made can raise the accessible information by
an arbitrary amount, depending on the system size.

This throws into sharp contrast two security criteria for
quantum cryptography: the ``na\"{i}ve'' one, which only asks for
the eavesdropper to have small accessible information about
the key, and the ``correct'', composable one, which demands that the
quantum mutual information is small~\cite{KRBM-locking}. In
fact, in~\cite{KRBM-locking} it was shown that this choice
can make a big difference: a large key may appear private 
according the former criterion, but not under the latter. 

Quantifying this difference, Guha \emph{et al.}~have recently
introduced the notion of locking capacity of a channel~\cite{Guha-et-al:enigma},
following Lloyd's suggestion of ``quantum enigma machines''~\cite{Lloyd-aenigma};
actually two capacities, 
one \emph{strong}, and one \emph{weak locking capacity} 
of a channel. Here we will only look at the weak variant,
which is the largest rate of asymptotically reliable classical communication
between the ``legal'' users (Alice and Bob), such that the accessible 
information of the eavesdropper observing the channel environment
(complementary channel output), about a uniformly distributed message,
goes to zero.

\medskip
To be precise, let Alice and Bob be connected by a quantum channel,
\emph{i.e.}~a completely positive and trace preserving (cptp) map
$\cN:\cL(A) \longrightarrow \cL(B)$, with (here: finite dimensional)
Hilbert spaces $A$ and $B$. It has a Stinespring dilation, via an
essentially unique isometry $V:A \hookrightarrow B \ox E$, where
$E$ is the eavesdropper's system (Eve):
$\cN(\rho) = \tr_E V\rho V^\dagger$. Tracing over $B$ instead yields
the \emph{complementary channel} $\cN^c(\rho) = \tr_B V\rho V^\dagger$
from Alice to Eve. All of our discussion of privacy will be in this
model, which is a quantum version of Wyner's wiretap channel~\cite{Wyner:wiretap}.

To communicate via $n$ instances of the channel, Alice and Bob employ
an \emph{$(n,\epsilon)$-code}, which is a collection 
$\{(\rho_m,D_m):m=1,\ldots,N\}$ consisting of states $\rho_m$ on $A^n$
and POVM elements $D_m$ on $B^n$ (\emph{i.e.}~$D_m \geq 0$, $\sum_m D_m = \1$),
with the property that the estimate $\widehat{m}$ of $m$ obtained by
measuring $(D_m)$ on the channel output is very likely to equal $m$,
which is assumed to be drawn uniformly:
\begin{equation}
  \label{eq:n-epsilon-code}
  P_{\text{err}} = \Pr\{M\neq\widehat{M}\}
                 = \frac{1}{N} \sum_{m=1}^N \tr\bigl( \cN^{\ox n}(\rho_m) (\1-D_m) \bigr)
                 \leq \epsilon.
\end{equation}
Given a code, we call it \emph{$\delta$-private} (for the channel $\cN$)
if there exists a state $\omega_0$ on $E^n$ such that
\begin{equation}
  \label{eq:delta-private}
  \frac{1}{N} \sum_{m=1}^N \bigl\| (\cN^c)^{\ox n}(\rho_m) - \omega_0 \bigr\|_1 \leq \delta.
\end{equation}
This condition captures precisely the commonly accepted notion of
private communication, since it says that Eve's output is typically
close to a constant, independent of the message $m$. Strictly speaking,
the above notion is that of a \emph{secret key generation code},
since we impose an a priori uniform distribution on the $m$'s.

Finally, the code is called \emph{$\delta$-weakly-locked} (always for the 
same channel $\cN$), if for every POVM $(Q_j)$ on $E^n$ there exists a probability 
distribution $\Omega = (\Omega_j)$ such that
\begin{equation}
  \label{eq:delta-weakly-locked}
  \frac{1}{N} \sum_{m=1}^N \sum_j \left| \tr \bigl[(\cN^c)^{\ox n}(\rho_m)Q_j\bigr] - \Omega_j \right| \leq \delta.
\end{equation}
By the contractive property of the trace norm under cptp maps (in this
case $\sigma \mapsto \sum_j \proj{j} \tr\sigma Q_j$), the $\delta$-private
property implies $\delta$-weak-locking. Guha \emph{et al.}~\cite{Guha-et-al:enigma}
have also defined the notion of \emph{strong locking}, which boils down to the
set of signal states $\rho_m$ satisfying eq.~(\ref{eq:delta-weakly-locked})
for the identity channel $\cN^c$, \emph{i.e.}~constant channel $\cN$:
\begin{equation}
  \label{eq:delta-strongly-locked}
  \frac{1}{N} \sum_{m=1}^N \sum_j \bigl| \tr \rho_m Q_j - \Omega_j \bigr| \leq \delta.
\end{equation}
However, we shall not prove any new results on strongly locked codes,
and include the definition only for completeness, see however recent
progress in~\cite{LupoLloyd}.

Following~\cite{Guha-et-al:enigma}, the code in fact may depend on 
a sublinear secret key $k$ (\emph{i.e.}~of $o(n)$ bits) pre-shared between 
Alice and Bob, so that in eq.~(\ref{eq:n-epsilon-code}) the states $\rho_{mk}$
and POVMs elements $D_{mk}$ depend on $m$ and $k$, and the error probability
includes also an average over $k$. On the other hand, in the privacy
and weak-locking conditions, eqs.~(\ref{eq:delta-private}) and (\ref{eq:delta-weakly-locked}),
we have to put $\rho_m = \EE_k \rho_{mk}$, the average over the key $k$, 
because it is unknown to Eve. 
An equivalent way of including the pre-shared key, which we
prefer here as it allows us to keep the above definitions of wiretap 
channel codes, is to grant Alice and Bob the use of $o(n)$ instances 
of an ideal qubit channel  (which automatically is perfectly private) 
in addition to the $n$ instances of $\cN$. Then all we have to do is to 
substitute $\id_2^{\ox o(n)}\!\ox\,\cN^{\ox n}$
for the main channel in eqs.~(\ref{eq:n-epsilon-code}), (\ref{eq:delta-private})
and (\ref{eq:delta-weakly-locked}) above.

\medskip
With these notions, we can give the definitions of channel capacities
as the largest asymptotic rate $R = \frac1n \log N$ attainable with arbitrarily
small error:
\begin{align*}
  P(\cN)   &:= \sup \left\{ R : \exists\ \delta\text{-private } (n,\epsilon)\text{-codes with } 
                                          N \geq 2^{nR},\ \epsilon,\delta \rightarrow 0 \right\}, \\
  L_W(\cN) &:= \sup \left\{ R : \exists\ \delta\text{-weakly-locked } (n,\epsilon)\text{-codes with } 
                                          N \geq 2^{nR},\ \epsilon,\delta \rightarrow 0 \right\}, \\
  L_S(\cN) &:= \sup \left\{ R : \exists\ \delta\text{-strongly-locked } (n,\epsilon)\text{-codes with } 
                                          N \geq 2^{nR},\ \epsilon,\delta \rightarrow 0 \right\},
\end{align*}
are the private, weak locking and strong locking capacity, respectively.
By definition, $L_S(\cN) \leq L_W(\cN)$ and $P(\cN) \leq L_W(\cN) \leq C(\cN)$,
the latter being the classical capacity of $\cN$. It can be shown that
the quantum capacity $Q(\cN)$ is a lower bound on $L_S(\cN)$, but
the relation between $P(\cN)$ and $L_S(\cN)$ is unknown~\cite{Guha-et-al:enigma}.

In cryptographic contexts, we would also worry about the speed of 
convergence of $\epsilon$ and $\delta$, usually by introducing exponential
decay rates, $\epsilon = 2^{-nE}$, $\delta = 2^{-nS}$ ($S>0$ is
called a \emph{security parameter}), in which case we would have to
study the tradeoff between rate $R$ and the error/security rates $E$ and $S$.
The private capacity $P(\cN)$ has been determined in~\cite{Devetak:P+Q,CaiWinterYeung:P}
and from the proof we know that by letting $E>0$ and $S>0$ sufficiently
small, rates arbitrarily close to $P(\cN)$ can be achieved. A priori
this is not clear for the locking capacities, although the results presented
in this paper show that at least certain weak locking rates, sometimes 
even rates arbitrarily close to $L_W(\cN)$, can be achieved with 
$\epsilon,\,\delta = 2^{-\Omega(n)}$.

\begin{remark}
  \normalfont 
  Guha \emph{et al.}~\cite[Def.~1]{Guha-et-al:enigma} give a very similar
  definition of locking, but demand the much stronger condition that
  in eq.~(\ref{eq:delta-weakly-locked}) the \emph{conditional}
  distribution of $m$ given each outcome $j$ has to be close to
  uniform. This however seems too restrictive, and not in line with the
  usual modern definition of privacy in wiretap 
  channels~\cite{Devetak:P+Q,CaiWinterYeung:P}, reflected in 
  eq.~(\ref{eq:delta-private}); in fact, that definition would assign
  a private capacity of zero to the perfectly innocent and well-understood
  quantum erasure channel~\cite{Smith}.
  Thus we propose to use our criterion (\ref{eq:delta-weakly-locked}).
  
  Guha \emph{et al.}~\cite{Guha-et-al:enigma}
  also discuss the possibility of defining the weak locking 
  property in terms of the Shannon mutual information between $m$ and 
  $j$ in eq.~(\ref{eq:delta-weakly-locked}), the maximum of which over
  all measurements is the \emph{accessible information} 
  \[
    I_{\text{acc}}(M:E^n) = I_{\text{acc}}\left(\left\{\frac1M,(\cN^c)^{\ox n}(\rho_m)\right\}\right)
  \]
  of the uniform ensemble of the eavesdropper's output states.
  Likewise, the privacy of a code could also have been characterized
  in terms of the quantum mutual information $I(M:E^n)$, which equals 
  the Holevo information of the ensemble 
  $\left\{\frac1M,(\cN^c)^{\ox n}(\rho_m)\right\}$. 
  For a generic ensemble $\cE = \{p_x,\sigma_x\}$ of states on $A$, 
  and corresponding cq-state $\sum_x p_x \proj{x}^X \ox \sigma_x^A$, 
  these information quantities are defined as 
  \begin{align*}
    I(X:A) &= S(X)+S(A)-S(XA) = S(A)-S(A|X) \\
           &= \chi(\cE) = S\left(\sum_x p_x\sigma_x\right) - \sum_x p_x S(\sigma_x), \\
    I_{\text{acc}}(X:A)
           &= I_{\text{acc}}(\cE)
            = \max_{\text{POVM }(Q_y)} I(X:Y) \ \text{ with }\ \Pr\{X=x,Y=y\} = p_x\tr\sigma_x Q_y. 
  \end{align*}
  
  By the Alicki-Fannes inequality~\cite{Alicki-Fannes}, a $\delta$-private
  code satisfies $I(M:E^n) \leq O(n)\delta$, and likewise a $\delta$-weakly-locking
  code satisfies $I_{\text{acc}}(M:E^n) \leq O(n)\delta$.
  Vice versa, Pinsker's inequality %\textcolor{red}{(CITE?)}
  implies that $I(M:E^n) \leq \Delta$
  and $I_{\text{acc}}(M:E^n) \leq \Delta$ imply $\sqrt{2\Delta}$-privacy
  and $\sqrt{2\Delta}$-weak-locking, respectively (similarly for strong locking). 
  Hence, as long as $\delta$ in our definitions above is $o(1/n)$, 
  the resulting notions of weak and strong locking, as well as private, capacity,
  are equivalent to the present ones.
\end{remark}

\medskip
In the present paper, 
we shall take a closer look at the weak locking capacity for
so-called \emph{degradable} channels $\cN$, which means that
there is a cptp map $\cD:\cL(B) \longrightarrow \cL(E)$ satisfying
$\cN^c = \cD \circ \cN$. We call $\cN$ \emph{anti-degradable} iff
the complementary channel $\cN^c$ is degradable. If a channel
$\cN$ is both degradable and anti-degradable, and specifically
if the degrading map $\cD$ is an isomorphism between $B$ and $E$, 
we call it \emph{symmetric}.

\begin{remark}
  \normalfont
  For degradable channels $\cN$, it is well-known~\cite{Devetak-Shor,Smith} 
  that
  \[
    P(\cN) = Q(\cN) = \max_{\rho} S(\cN(\rho)) - S(\cN^c(\rho)),
  \]
  where the right hand side is the maximization of the coherent
  information, which is concave in $\rho$.
  By definition, any private communication code is a weak locking code
  for $\cN$, hence $L_W(\cN) \geq P(\cN)$.
  
  For an anti-degradable channel $\cN$, $Q(\cN) = P(\cN) = 0$
  by the familiar ``cloning argument''~\cite{BDSW:bigpaper}. 
  Below we will see that for the weak locking capacity this does not hold.
\end{remark}

\medskip
In~\cite{Guha-et-al:enigma} it had been left open whether the weak
locking capacity is always equal to the private capacity, or whether
there can be a separation.
For example, there it was shown that for channels
$\cN$ such that the complementary channel $\cN^c$ is a qc-channel, 
then $L_W(\cN)=P(\cN)$; furthermore, that if $\cN$ is entanglement-breaking, 
then $L_W(\cN)=P(\cN)=0$.
Note that the construction in~\cite{KRBM-locking} (as well 
as~\cite{ChristandlEHHOR}) may be taken as 
evidence for large gaps, but it is not sufficient to prove this:
Namely, in those papers it was pointed out that \emph{if} at the end
of a hypothetical key agreement protocol Alice and Bob share perfect
randomness, and their correlation with Eve is described as
\[
  \frac1N \sum_{m=1}^N \proj{m}^A \ox \proj{m}^B \ox \rho_m^E,
\]
with a strongly-locking ensemble $\left\{\frac1N,\rho_m\right\}$,
cf.~eq.~(\ref{eq:delta-strongly-locked}), then the key may not
be secure at all after a small portion ($\ll \log N$) of the 
shared secret has been leaked. Our contribution is to show that 
this can indeed occur naturally in the above outlined setting of 
the quantum wiretap channel.

Here we show a general lower bound on $L_W(\cN)$ for channels
$\cN$ such that the complementary channel $\cN^c$ is a cq-channel
(these are automatically degradable);
we establish basic properties of these channels, including an upper bound
on $L_W(\cN)$, in the next Section~\ref{sec:cq-channels}.
This bound can sometimes be much larger than $P$ (Section~\ref{sec:main}). 
We also exhibit symmetric channels, hence with vanishing private
capacity, which nonetheless have positive weak locking capacity
(Section~\ref{sec:P=0}). After this we conclude with a discussion
of our results in the context of regular quantum key distribution (QKD)
and several open questions, in Section~\ref{sec:further}.

\section{Complements of cq-channels}
\label{sec:cq-channels}
One subclass we'll be interested in are so-called \emph{Hadamard channels}~\cite{hadamard},
specifically those that are complementary channels of 
\emph{cq-channels}~\cite{Holevo:1977}:
\begin{equation}
  \label{eq:cq-channel}
  \cN^c(\ketbra{i}{i'}) = \delta_{ii'} \rho_i^E,
\end{equation}
where $\{ \ket{i} \}$ is an orthonormal basis of $A$, and
with $\rho_i^E = \tr_B \proj{\psi_i}^{BE}$, so that
\begin{equation}
  \label{eq:our-channel}
  \cN(\ketbra{i}{i'}) = \ketbra{i}{i'} \otimes \tr_E \ketbra{\psi_i}{\psi_{i'}}.
\end{equation}
Note that in general, Hadamard channels are defined as complementary
channels of entanglement-breaking channels, which results in a wider
class than the ones we are looking at here~\cite{hadamard}. The more
restrictive class of channels in eq.~(\ref{eq:our-channel}) are also known 
as \emph{Schur multipliers}.

The cq-channels are called so, because they are ``classical-to-quantum''~\cite{Holevo:1977}. 
The opposite concept of \emph{qc-channel} (``quantum-to-classical'')
models a measurement as a cptp map; for a POVM $(Q_j)$, it is given by
\begin{equation}
  \label{eq:qc-channel}
  \cN(\rho) = \sum_j \tr\rho Q_j \proj{j}.
\end{equation}

We begin with an upper bound on the weak locking capacity,
to have a benchmark for our lower bound later on.

\begin{proposition}
  \label{prop:upper-bound}
  Let $\cN:\cL(A)\longrightarrow\cL(B)$ be a Schur multiplier, i.e.~a Hadamard channel whose
  complementary channel $\cN^c:\cL(A)\longrightarrow\cL(E)$ is a cq-channel. 
  Then,
  \begin{align*}
    L_W(\cN)            &\leq \max_{(p_i)} S_{\text{acc}}(I|E), \text{ where}\\ 
    S_{\text{acc}}(I|E) &:=   \min_{(Q_j)} H(I|J),
  \end{align*}
  is the \emph{eavesdropper's accessible equivocation}. Here, $(p_i)$ is a 
  probability distribution on the computational basis states of $A$ and
  $(Q_j)$ is a POVM on $E$, $\Pr\{I=i,J=j\} = p_i\tr\rho_i Q_j$.
\end{proposition}
\begin{proof}
Basically, we evaluate the upper bound from~\cite[Thm.~8]{Guha-et-al:enigma}:
$L_W(\cN) \leq \sup_n \frac1n L_W^{(u)}(\cN^{\ox n})$, where
\begin{equation}
  \label{eq:L_W_upper}
  L_W^{(u)}(\cN) = \max_{\{p_x,\rho_x\}} I(X:B) - I_{\text{acc}}(X:E)
\end{equation}
is optimized with respect to arbitrary ensembles $\{p_x,\rho_x\}$ of
states on $A$. 

Choosing any probability distribution $(p_i)$ on the computational
basis states $\rho_i = \proj{i}$, we get $I(I:B) = H(I)$ and hence
$L_W^{(u)}(\cN) \geq S_{\text{acc}}(I|E)$.
Furthermore, for the tensor product $\cN_1\ox\cN_2$ of two complements 
of cq-channels,
\[
  \max_{(p_{i_1i_2})} S_{\text{acc}}(I_1 I_2|E_1 E_2)
      = \max_{(p_{i_1})} S_{\text{acc}}(I_1|E_1) + \max_{(p_{i_2})} S_{\text{acc}}(I_2|E_2),
\]
by Lemma~\ref{lemma:S-acc-add} below. This shows in fact that for
any integer $n$,
$\frac1n L_W^{(u)}\bigl(\cN^{\ox n}\bigr) \geq \max_{(p_i)} S_{\text{acc}}(I|E)$.

Thus, it remains to show $L_W^{(u)}(\cN) \leq \max_{(p_i)} S_{\text{acc}}(I|E)$. 
To do so, we shall 
first show that for degradable channels, an optimal ensemble for 
eq.~(\ref{eq:L_W_upper}) consists w.l.o.g.~of pure states 
$\rho_x = \proj{\varphi_x}$, 
and then in a second step that we can choose these pure states as
computational basis states, modifying the ensemble accordingly.

\emph{1.} For the degradable channel $\cN$, choose a Stinespring isometry
$V_0:A \hookrightarrow B \ox E$, and for the degrading map an isometry
$V_1:B \hookrightarrow E' \ox F$. The accessible information
requires a measurement $(Q_j)$ -- w.l.o.g.~consisting of rank-one
operators --, for whose associated qc-channel we choose an isometry
(acting on $E'$ but of course equally on $E$)
$V_2:E' \hookrightarrow J \ox J'$. Now, given an ensemble $\cE = \{p_x,\rho_x\}$,
\begin{equation}\begin{split}
  \label{eq:CK-expression}
  I(X:B) - I_{\text{acc}}(X:E) &= I(X:FJJ') - I(X:J)    \\
                               &= I(X:FJ'|J)            \\
                               &= H(FJ'|J) - H(FJ'|JX),
\end{split}\end{equation}
where all expressions except the l.h.s.~are with respect to the state
\[
  \omega^{XFJJ'} = \sum_x p_x \proj{x}^X \ox \bigl( V_2 V_1 \cN(\rho_x) V_1^\dagger V_2^\dagger \bigr)^{FJJ'}.
\]
On the r.h.s.~of eq.~(\ref{eq:CK-expression}), $H(FJ'|J)$ depends only
on $\omega^{FJJ'}$ and so is unchanged if we replace each $\rho_x$
in $\cE$ by any of its pure state decompositions. On the other hand,
\[
  H(FJ'|JX)_\omega = \sum_x p_x H(FJ'|J)_{V_2 V_1 \cN(\rho_x) V_1^\dagger V_2^\dagger},
\]
and since the conditional entropy is concave in the state~\cite{SSA},
this replacement can make the latter quantity only smaller.

\emph{2.} Now that we know that we may assume a pure state ensemble 
$\cE = \{p_x,\rho_x=\proj{\varphi_x}\}$, we specialize to the complements
of cq-channels. Looking at eqs.~(\ref{eq:cq-channel}) and (\ref{eq:our-channel}),
we see that $\cN^c$ is invariant, and $\cN$ covariant, under conjugation
by phase (diagonal) unitaries. By twirling the ensemble by phase unitaries
(\emph{i.e.}~replacing each $\ket{\varphi_x}$ by a uniform distribution over
$U^{\text{diag}}\ket{\varphi_x}$), we thus can only increase the r.h.s.~of
eq.~(\ref{eq:CK-expression}) by leaving $H(FJ'|JX)$ alone while
$H(FJ'|J)$ can only increase, since $\omega^{FJJ'}$ is now invariant under 
conjugation by phase unitaries.

The proof will be concluded by showing that
\[
  H(FJ'|JX) = \sum_x p_x H(FJ'|J)_{V_2 V_1 \cN(\proj{\varphi_x}) V_1^\dagger V_2^\dagger}
\]
can only decrease if we replace each $\ket{\varphi_x} = \sum_i \alpha_{i|x}\ket{i}$ 
by the ensemble $\{ p_{i|x} = |\alpha_{i|x}|^2,\proj{i} \}$, hence the 
original ensemble $\cE$ by $\widetilde{\cE} = \{ p_{xi} = p_x p_{i|x}, \proj{i} \}$, which in turn
has the same value of the expression (\ref{eq:CK-expression}) as 
$\{ p_i =\sum_x p_{xi}, \proj{i} \}$.
Indeed, the corresponding $\widetilde{\omega}$ has the same reduction
on $FJJ'$, $\omega^{FJJ'} = \widetilde{\omega}^{FJJ'}$, and
for every $x$, we have
\begin{equation}
  \label{eq:final-step}
  H(FJ'|J)_{V_2 V_1 \cN(\proj{\varphi_x}) V_1^\dagger V_2^\dagger} 
     \geq \sum_i p_{i|x} H(FJ'|J)_{V_2 V_1 \cN(\proj{i}) V_1^\dagger V_2^\dagger}.
\end{equation}
To see this, we expand the l.h.s.~as
\[\begin{split}
  H(FJ'|J)_{V_2 V_1 \cN(\proj{\varphi_x}) V_1^\dagger V_2^\dagger} 
     &= H(FJJ')_{V_2 V_1 \cN(\proj{\varphi_x}) V_1^\dagger V_2^\dagger} 
        - H(J)_{V_2 V_1 \cN(\proj{\varphi_x}) V_1^\dagger V_2^\dagger}                         \\
     &= S\bigl(\cN(\proj{\varphi_x})\bigr) - H\bigl(\{\tr\cN^c(\proj{\varphi_x}) Q_j\}_j\bigr) \\
     &= S\bigl(\cN^c(\proj{\varphi_x})\bigr) - H\bigl(\{\tr\cN^c(\proj{\varphi_x}) Q_j\}_j\bigr),
\end{split}\]
and observe
$\cN^c(\proj{\varphi_x}) = \sum_i p_{i|x}\cN^c(\proj{i}) = \sum_i p_{i|x}\rho_i =: \overline{\rho}$.
While on the r.h.s., for each $i$,
\[\begin{split}
  H(FJ'|J)_{V_2 V_1 \cN(\proj{i}) V_1^\dagger V_2^\dagger} 
     &= H(FJJ')_{V_2 V_1 \cN(\proj{i}) V_1^\dagger V_2^\dagger} 
        - H(J)_{V_2 V_1 \cN(\proj{i}) V_1^\dagger V_2^\dagger}               \\
     &= S\bigl(\cN(\proj{i})\bigr) - H\bigl(\{\tr\cN^c(\proj{i}) Q_j\}_j\bigr) \\
     &= S\bigl(\cN^c(\proj{i})\bigr) - H\bigl(\{\tr\cN^c(\proj{i}) Q_j\}_j\bigr) \\
     &= S\bigl(\rho_i\bigr) - H\bigl(\{\tr\rho_i Q_j\}_j\bigr).
\end{split}\]
Hence, the difference between l.h.s.~and r.h.s.~of eq.~(\ref{eq:final-step}) is
\[\begin{split}
  H(FJ'|J)_{V_2 V_1 \cN(\proj{\varphi_x}) V_1^\dagger V_2^\dagger} 
   &- \sum_i p_{i|x} H(FJ'|J)_{V_2 V_1 \cN(\proj{i}) V_1^\dagger V_2^\dagger} \\
   &\!\!\!\!\!\!\!\!\!\!
    = S(\overline{\rho}) - \sum_i p_{i|x} S(\rho_i)
        - H\bigl(\{\tr\overline{\rho}Q_j\}_j\bigr) + \sum_i p_{i|x} H\bigl(\{\tr\proj{i}Q_j\}_j\bigr) \\
   &\!\!\!\!\!\!\!\!\!\!
    = I(I:E) - I(I:J) \geq 0,
\end{split}\]
the last inequality by the famous Holevo bound~\cite{Holevo:bound}, and we are done.
\end{proof}

\begin{lemma}
  \label{lemma:S-acc-add}
  For the tensor product of channels $\cN_1$ and $\cN_2$, each of 
  which is the complement of a cq-channel,
  \[
    \max_{(p_{i_1i_2})} S_{\text{acc}}(I_1 I_2|E_1 E_2)
       = \max_{(p_{i_1})} S_{\text{acc}}(I_1|E_1) + \max_{(p_{i_2})} S_{\text{acc}}(I_2|E_2).
  \]
\end{lemma}
\begin{proof}
First, for any distribution $(p_{i_1i_2})$, and any measurement POVM $(Q_j)$,
we have, by subadditivity of the entropy,
$H(I_1 I_2|J) \leq H(I_1|J) + H(I_2|J)$, hence, choosing the POVM to be a
tensor product of local POVMs, $Q_{j_1 j_2} = Q_{j_1} \ox Q_{j_2}$, we get
\[
  S_{\text{acc}}(I_1 I_2|E_1 E_2)
     \leq S_{\text{acc}}(I_1|E_1) + S_{\text{acc}}(I_2|E_2).
\]

On the other hand, consider a product distribution $p_{i_1 i_2}=p_{i_1}p_{i_2}$,
the output of $(\cN_1\ox\cN_2)^c$ is a product ensemble
$\{p_{i_1},\rho_{i_1}\} \ox \{p_{i_2},\rho_{i_2}\}$. For a generic
POVM $(Q_j)$ on $E_1E_2$, we can switch around the roles of
the ensemble and of the POVM, observing that with 
$\overline{\rho}^{(b)} = \sum_{i_b} p_{i_b}\rho_{i_b}$ 
and the POVMs(!) composed of the operators
$M_{i_b} = \left(\overline{\rho}^{(b)}\right)^{-\frac12} p_{i_b}\rho_{i_b} 
           \left(\overline{\rho}^{(b)}\right)^{-\frac12}$ ($b=0,1$),
\[\begin{split}
  \Pr\{I_1=i_1,I_2=i_2,J=j\} &= p_{i_1}p_{i_2}\tr(\rho_{i_1}\ox\rho_{i_2})Q_j \\
                             &= \tr\left(\sqrt{\overline{\rho}^{(1)}\ox\overline{\rho}^{(2)}} 
                                           Q_j 
                                         \sqrt{\overline{\rho}^{(1)}\ox\overline{\rho}^{(2)}}\right)
                                        (M_{i_1}\ox M_{i_2})                  \\
                             &=: q_j \tr \sigma_j(M_{i_1}\ox M_{i_2}).
\end{split}\] 
Thus,
\[
  H(I_1I_2|J) = \sum_j q_j S\bigl( (\cM_1\ox\cM_2)\sigma_j \bigr),
\]
where $\cM_b$ is the qc-channel representing the POVM $(M_{i_b})$ ($b=0,1$).
This means
\[
  S_{\text{acc}}(I_1I_2|E_1E_2) 
        = \widehat{H}\bigl(\cM_1\ox\cM_2|\overline{\rho}^{(1)}\ox\overline{\rho}^{(2)}\bigr),
\]
and likewise
\[
  S_{\text{acc}}(I_1|E_1) = \widehat{H}\bigl(\cM_1|\overline{\rho}^{(1)}\bigr),
    \quad
  S_{\text{acc}}(I_2|E_2) = \widehat{H}\bigl(\cM_2|\overline{\rho}^{(2)}\bigr),
\]
where 
\[
  \widehat{H}(\cM|\sigma) := \min_{\{q_j,\psi_j\}} \sum_j q_j H\bigl(\cM(\psi_j)\bigr)
                                                   \quad\text{s.t.}\quad \sum_j q_j \psi_j = \sigma
\]
is the \emph{constrained minimum output entropy}.
But now we can invoke~\cite[Lemma 3]{King:additive}, by which
\[
  \widehat{H}\bigl(\cM_1\ox\cM_2|\overline{\rho}^{(1)}\ox\overline{\rho}^{(2)}\bigr)
      = \widehat{H}\bigl(\cM_1|\overline{\rho}^{(1)}\bigr) 
        + \widehat{H}\bigl(\cM_2|\overline{\rho}^{(2)}\bigr),
\]
concluding the proof.
\end{proof}

\medskip
Even though we do not make use of it here, we cannot pass without noting
the following fundamental property of the optimization of $H(I|J)$:
\begin{lemma}
  \label{lemma:minimax}
  For an ensemble $\cE = \{p_i,\rho_i\}$ and a POVM $Q = (Q_j)$,
  the function $\eta(\cE;Q) = H(I|J)$ is concave in $\cE$ and
  convex (actually affine) in $Q$, in the following sense:
  For ensembles $\cE^{(0)} = \{p_i^{(0)},\rho_i\}$ and
  $\cE^{(1)} = \{p_i^{(1)},\rho_i\}$ (without loss of generality
  sharing the same set of states),
  $\cE = \lambda\cE^{(0)} + (1-\lambda)\cE^{(1)} 
       = \{\lambda p_i^{(0)} + (1-\lambda) p_i^{(1)},\rho_i\}$
  satisfies
  \[
    \eta(\lambda\cE^{(0)} + (1-\lambda)\cE^{(1)},Q) 
          \geq \lambda \eta(\cE^{(0)},Q) + (1-\lambda) \eta(\cE^{(1)},Q).
  \]
  Instead, for POVMs $(Q_j^{(0)})$ and $(Q_k^{(1)})$ on \emph{disjoint}
  index sets $\{j\}$ and $\{k\}$,
  $Q = \lambda Q^{(0)} \oplus (1-\lambda)Q^{(1)}
     = \bigl( \lambda Q_j^{(0)},(1-\lambda) Q_k^{(1)} \bigr)$
  satisfies
  \[
    \eta(\cE,\lambda Q^{(0)} \oplus (1-\lambda)Q^{(1)})
           =   \lambda \eta(\cE,Q^{(0)}) + (1-\lambda) \eta(\cE,Q^{(1)}).
  \]
  Consequently,
  \[
    \max_{(p_i)} S_{\text{acc}}(I|E) = \max_{(p_i)} \min_{(Q_j)} H(I|J)
                                     = \min_{(Q_j)} \max_{(p_i)} H(I|J).
  \]
\end{lemma}
\begin{proof}
The concavity property boils down to the concavity of the
Shannon entropy. The affine-linearity is evident from the definition.
Finally, the minimax statement is an application of the concavity
in the first and convexity in the second argument, invoking
von Neumann's minimax theorem~\cite{vonNeumann:minimax,Sion:minimax}.
\end{proof}

\section{Lower bound on $\mathbf{L_W}$ for complements of cq-channels}
\label{sec:main}
We shall need the following auxiliary lemma, which was proved by
Damgaard \emph{et al.}~\cite{Damgaard:high-order} in a very similar form.
In Appendix~\ref{app:simple} we give a simple proof of it, based
on the additivity of the minimum output R\'{e}nyi-entropy for 
entanglement-breaking channels~\cite{King:additive}.

\begin{proposition}
  \label{prop:high-order}
  Consider a POVM $M=(M_i)$ and its associated qc-channel $\cM$, with
  minimum output entropy 
  $\displaystyle{\widehat{H}(\cM) := \min_{\psi\text{ state}} H(\cM(\psi))}$,
  where $H(\cM(\psi)) = H\bigl(\{q_i = \tr \psi M_i\}\bigr)$.
  Then, for any $0< \epsilon,\,\delta < 1$, and any state $\psi$ on $n$
  input systems,
  \[
    H_{\min}^\epsilon\bigl( \cM^{\ox n}(\psi) \bigr)
                     \geq n\bigl( \widehat{H}(\cM) - \delta \bigr)
                            - {16(\log d)^2} \frac{1}{\delta} \log\frac{1}{\epsilon}.
  \]
  For a suitable choice of $\delta$ (depending on $\epsilon$), 
  we get (for sufficiently large $n$, ensuring that $\delta < 1$):
  \[
    H_{\min}^\epsilon\bigl( \cM^{\ox n}(\psi) \bigr)
                     \geq n \widehat{H}(\cM) - 8(\log d) \sqrt{n \log\frac{1}{\epsilon}}.
  \]
\end{proposition}

Here, $H_{\min}^\epsilon$ is the \emph{smooth min-entropy}~\cite{Renner:PhD}:
\begin{definition}
  \label{defi:H_min}
  For a state $\rho$, the \emph{min-entropy} is $H_{\min}(\rho) := -\log \|\rho\|$,
  and the \emph{smooth min-entropy}
  \[
    H_{\min}^\epsilon(\rho) = \max H_{\min}(\rho') 
            \quad\text{s.t.}\quad \frac12 \|\rho-\rho'\|_1 \leq \epsilon.
  \]
  More generally, for a bipartite state $\rho^{AB}$,
  \[\begin{split}
    H_{\min}(A|B)_\rho &:=   -\log \min \lambda \quad\text{s.t.}\quad 
                                   \rho^{AB} \leq \lambda(\1^A\ox\sigma^B),\ \sigma \text{ state} \\
                       &\geq -\log \left\| (\1\ox\rho^B)^{-1/2} \rho (\1\ox\rho^B)^{-1/2} \right\|
                              =: H_{\infty}(A|B)_\rho,
  \end{split}\]
  and
  \[
    H_{\min}^\epsilon(A|B)_\rho := \max H_{\min}(A|B)_{\rho'} 
                                      \quad\text{s.t.}\quad \frac12 \|\rho-\rho'\|_1 \leq \epsilon.
  \]
\end{definition}

\begin{remark}
  \label{rem:min-entropy}
  \normalfont
  Unlike the nowadays standard definition of smooth (conditional) min-entropy,
  which uses the so-called \emph{purified distance}~\cite{Tomamichel:PhD},
  we employ the trace distance. This is essentially equivalent, since
  the two metrics are dominating each other (in fact, trace distance is
  upper bounded by the purified distance). However, it makes for more
  direct application of classical randomness extraction results later.
  
  Note that for a qc-state $\rho^{AB} = \sum_j q_j \rho_j^A \ox \proj{j}^B$,
  \begin{align*}
    H_{\min}(A|B)_\rho                   &\geq \min_j H_{\min}(\rho_j), \\
    H_{\min}^{\epsilon+\delta}(A|B)_\rho &\geq \min_{j\in\cT} H_{\min}^\epsilon(\rho_j),
  \end{align*}
  for any set $\cT$ of indices with $\Pr\{j\not\in\cT\} \leq \delta$.
\end{remark}

\begin{corollary}
  \label{cor:accessible-conditional-Hmin}
  Let $\{p_i,\rho_i\}$ be an ensemble on a Hilbert space $E$ with
  associated average state $\overline{\rho}$ and POVM 
  $(M_i = \overline{\rho}^{-\frac12}p_i\rho_i\overline{\rho}^{-\frac12})$.
  Then, for any POVM $Q=(Q_j)$ on $E^n$, i.i.d.~$I_1,\ldots,I_n \sim (p_i)$
  and $0 < \delta < 1$,
  \[
    H_{\min}^\epsilon(I^n|J) \geq n\bigl( \widehat{H}(\cM) - \delta \bigr)
                                           - {16(\log d)^2} \frac{1}{\delta} \log\frac{1}{\epsilon},
  \]
  and for sufficiently large $n$,
  \[
    H_{\min}^\epsilon(I^n|J) \geq n \widehat{H}(\cM) - 8(\log d) \sqrt{n \log\frac{1}{\epsilon}}.
  \]
\end{corollary}
\begin{proof}
Simply use the trick to switch between ensembles and POVMs, to write
\[\begin{split}
  \Pr\{I^n=i^n,J=j\} &= p_{i_1}\cdots p_{i_n} \tr(\rho_{i_1}\ox\cdots\ox\rho_{i_n})Q_j \\
                     &= \tr\bigl(\sqrt{\overline{\rho}}^{\ox n} Q_j \sqrt{\overline{\rho}}^{\ox n}\bigr)
                            (M_{i_1}\ox\cdots\ox M_{i_n})                               \\
                     &= q_j \tr \sigma_j(M_{i_1}\ox\cdots\ox M_{i_n}).
\end{split}\]
Thus (cf.~Remark~\ref{rem:min-entropy}),
\[\begin{split}
  H_{\min}^\epsilon(I^n|J) &\geq \min_j H_{\min}^\epsilon(I^n|J=j) \\
                           &=    \min_j H_{\min}^\epsilon\bigl( \cM^{\ox n}(\sigma_j) \bigr) \\
                           &\geq \min_{\psi} H_{\min}^\epsilon\bigl( \cM^{\ox n}(\psi) \bigr), \\
\end{split}\]
and the claim follows from the lower bound of Proposition~\ref{prop:high-order}.
\end{proof}

\begin{theorem}
  \label{thm:main}
  For any Schur multiplier, i.e.~a Hadamard channel $\cN:\cL(A)\longrightarrow\cL(B)$ whose
  complementary channel $\cN^c:\cL(A)\longrightarrow\cL(E)$ is a cq-channel,
  $\cN^c(\ketbra{i}{j}) = \delta_{ij} \rho_i$, consider a distribution
  $(p_i)$ on the input computational basis states. Let 
  $(M_i = \overline{\rho}^{-\frac12} p_i\rho_i \overline{\rho}^{-\frac12})$
  be the POVM associated to the ensemble $\{p_i,\rho_i\}$ and denote
  the corresponding qc-channel $\cM:\cL(E) \longrightarrow \cL(A)$.
  Then,
  \[
    L_W(\cN) \geq \widehat{H}(\cM) = \min_{\psi\text{ state}} H(\cM(\psi)).
  \]
\end{theorem}
\begin{proof}
We first describe a secret key generation protocol in the sense of
weak locking: Alice generates i.i.d.~$I_1,\ldots,I_n \sim (p_i)$ and
sends the basis states $\ket{I_1}\cdots\ket{I_n}$ down the channel;
Bob, by the nature of the channel, receives these basis states
without noise, and so the string $I^n=I_1\ldots I_n$ serves as
a raw key shared between them.

Eve on the other hand, after measuring a POVM $Q$ on her output states
and obtaining outcomes $J$, has a certain min-entropy of $I^n$ given
$J$, which by Corollary~\ref{cor:accessible-conditional-Hmin} satisfies
\[
  H_{\min}^\epsilon(I^n|J) \geq n\bigl( \widehat{H}(\cM) - \delta \bigr)
                                 - {16(\log d)^2} \frac{1}{\delta} \log\frac{1}{\epsilon}.
\]
Thus, using a min-entropy extractor with $O(\log n)$ bits of ``seed''
randomness (which Alice and Bob are allowed as part of the sublinear
amount of key they may pre-share), they can convert almost all of
the smooth min-entropy into almost-uniform key $K$ that is almost-independent
of $J$; cf.~\cite[Section 6.2]{Vadhan:randomness} and references therein.
Mathematically, the extractor (more precisely: \emph{strong extractor})
is given by a function
$e:\cI^n \times \cS \longrightarrow \cK = \{0,1\}^{nR}$, 
$R = \widehat{H}(\cM) - 2\delta$ and $|\cS| = \text{poly}(n)$. 
It has the property that for every random variable $I^{(n)} \sim P^{(n)}$
on $\cI^n$ with min-entropy 
$\geq n(\widehat{H}(\cM) - \delta) - {16(\log d)^2} \frac{1}{\delta} \log\frac{1}{\epsilon}$
and uniformly distributed $S \in \cS$, $K = e(I^{(n)},S)$ is almost uniformly
distributed:
\begin{equation}
  \label{eq:extractor}
  \left\| \mathbb{P}(K,S) - \mathbb{U}_{\cK}\otimes\mathbb{U}_{\cS} \right\|_1 
                                                             \leq \frac{1}{\text{poly}(n)},
\end{equation}
where $\mathbb{U}_{\cK}\otimes\mathbb{U}_{\cS}$ is the uniform distribution
on $\cK \times \cS$.
This implies by triangle inequality, for Eve's measurement result $J$,
\begin{equation}
  \label{eq:weak-key}
  \left\| \mathbb{P}(J,K,S) 
           - \mathbb{P}(J) \otimes \mathbb{U}_{\cK} \otimes \mathbb{U}_{\cS} \right\|_1 
                                          \leq \eta := \epsilon + \frac{1}{\text{poly}(n)}.
\end{equation}
Observe that the bound $\frac{1}{\text{poly}(n)}$ comes from adding
the error terms $\frac{1}{\text{poly}(|\cS|)}$ and $2^{-n\Omega(\delta)}$
of the extractor~\cite{Vadhan:randomness}. 
Thus, making the seed space $\cS$ larger we can
suppress $\eta$ more, up to any quantity decaying to zero slower than
exponentially.

Now, to obtain a scheme to securely send uniformly distributed messages 
from $\cK$, we ``run the extractor backwards'': From the joint distribution
of $I^n$ [i.i.d.~according to $(p_i)$], $S$ (uniform) and $K = e(I^n,S)$ we 
can construct a conditional distribution 
$\mathbb{P}(I^n|K,S) =: E(i^n|k,s)$, which describes a stochastic encoding 
mapping $E: \cK \times \cS \longrightarrow \cI^n$. Note that we may
assume $I^n = E(K,S)$ as random variables.

To send the uniformly distributed message $\widehat{K} \in \cK$, Alice
and Bob share a uniformly distributed private $\widehat{S} \in \cS$, and Alice
puts $\widehat{I}^{(n)} = E(\widehat{K},\widehat{S}) \in \cI^n$.
We claim that this is a good code. Indeed, since he gets $\widehat{I}^{(n)}$
from the channel output, and using $\widehat{S}$, Bob can decode
$\widehat{K} = e(\widehat{I}^{(n)},\widehat{S})$ with certainty.

On the other hand, Eve can obtain almost no information about $\widehat{K}$,
because eq.~(\ref{eq:extractor}) means 
\[
  \left\| \mathbb{P}(K,S) - \mathbb{P}(\widehat{K},\widehat{S}) \right\|_1 \leq \eta,
\]
hence, applying the encoding map $E$,
\[
  \left\| \mathbb{P}(I^n,K,S) 
           - \mathbb{P}(\widehat{I}^{(n)},\widehat{K},\widehat{S}) \right\|_1 \leq \eta.
\]
Applying the complementary channel as well as Eve's POVM $Q$, we find
\[
  \left\| \mathbb{P}(J,K,S) - \mathbb{P}(\widehat{J},\widehat{K},\widehat{S}) \right\|_1 \leq \eta.
\]
Putting this together with eq.~(\ref{eq:weak-key}), and tracing out the seed,
we finally obtain
\[
  \left\| \mathbb{P}(\widehat{J},\widehat{K}) 
            - \mathbb{P}(J) \otimes \mathbb{U}_{\cK} \right\|_1 \leq 2 \eta,
\]
and letting $\epsilon$ and $\delta$ go to zero (slow enough) as
$n \rightarrow \infty$, we are done.
\end{proof}

\medskip
\begin{remark}
  \normalfont
  By using a seed of $o(n)$ bits in eq.~(\ref{eq:extractor}), and
  choosing $\epsilon = 2^{-o(n)}$ in eq.~(\ref{eq:weak-key}), we
  get $\eta$-weak-locking codes, with asymptotically the same rate
  and $\eta = 2^{-o(n)}$.
\end{remark}

\medskip
\begin{example}
  \label{example:MUB}
  \normalfont
  Consider $|E|=d$, $|A|=|B|=2d$ and the cq-channel $\cN^c$ with pure
  output states $\ket{v_{0i}} = \ket{i}$ and $\ket{v_{1i}} = \ket{\varphi_i}$ 
  ($i=1,\ldots,d$), which are the eigenstates of the generalized $Z$ and 
  $X$ operators, respectively.

  Using the concavity of the coherent information and the
  covariance of the channel under the action of the discrete Weyl group,
  it is easy to see that the coherent information is maximized 
  for the uniform input, and so
  \begin{equation}
    \label{eq:P=1}
    P(\cN) = 1.
  \end{equation}
  
  On the other hand, for uniform input distribution over the
  $2d$ basis states, the POVM 
  $\left(\frac12 M_{0i}\right) \cup \left(\frac12 M_{1i}\right)$, where
  $M_{bi} = \proj{v_{bi}}$, is the
  random choice of one of the observables $X$ or $Z$, and
  measurement of its eigenbasis. By the Maassen-Uffink entropic
  uncertainty relation~\cite{MaassenUffink},
  $\widehat{H}(\cM) = 1 + \frac12 \log d$, which by Theorem~\ref{thm:main}
  is a lower bound on $L_W(\cN)$. By Proposition~\ref{prop:upper-bound}
  it is also an upper bound, since regardless of the input distribution
  over the computational basis states $i0$, $i1$, Eve can randomly
  choose and measure either $X$ or $Z$, and get an accessible 
  equivocation of at most $1 + \frac12 \log d$.
  Hence,
  \begin{equation}
    \label{eq:L_W-big}
    L_W(\cN) = 1 + \frac12 \log d;
  \end{equation}
  and thus the gap between the private and (weak) locking capacities 
  of a $d$-dimensional channel 
  can be as large as a constant versus $\Omega(\log d)$.
  
  The Choi-Jamio\l{}kowski state obtained from using the above channel
  with maximally entangled input was previously considered 
  by Christandl \emph{et al.}~\cite[Sec.~6]{ChristandlEHHOR},
  finding the same numbers for the secret key rate as eqs.~(\ref{eq:P=1})
  and (\ref{eq:L_W-big}) as secret key rate against quantum and
  classical eavesdropper, respectively. Note however that for the
  latter conclusion they have to assume that Eve applies the same
  measurement to each copy of the shared state. The proof of
  Theorem~\ref{thm:main} shows that the conclusion of~\cite{ChristandlEHHOR}
  holds for arbitrary measurements of the $n$ systems in Eve's
  possession.
  \qed
\end{example}

\medskip
In the next section, we shall exhibit an example of an even more
striking effect: a channel whose private capacity, and indeed key
generation capacity, is zero, because Bob and Eve see the exact
same quantum information, but whose locking capacity is arbitrarily
large.

\section{Symmetric channels with $\mathbf{L_W > 0 = P}$}
\label{sec:P=0}
Compared to Section~\ref{sec:main}, are there even channels with vanishing private
capacity, \emph{i.e.}~$P(\cN) = 0$, but positive locking capacity, $L_W(\cN) > 0$?
Note that the construction in the previous section, to yield non-zero 
locking capacity,
requires a degradable Hadamard channel, and so its private capacity
is also non-zero. In this case, note also that the sublinear 
pre-shared key between Alice and Bob is unnecessary, since they can
use a sublinear number of channel uses and a private code to
create the key from scratch.

In this section we consider symmetric channels,
which trivially have $P = 0$; on the other hand, the pre-shared
key, even if only sublinear, can be enough of an advantage to
get a locking capacity.
For concreteness, let us look a little closer at the channel
\begin{equation}
  \label{eq:symmetric}
  \cS:\cL\bigl(\text{Sym}^2(B)\bigr) \longrightarrow \cL(B),
\end{equation}
with $B \simeq E \simeq \CC^d$,
which has as its Stinespring dilation the isometric embedding of
the $d\times d$ symmetric subspace $A = \text{Sym}^2(B) \simeq \CC^{d(d+1)/2}$
into $B \otimes E$.

One ``reasonable'' strategy to encode information is this:
Use the product states $\ket{\psi}\ket{\psi} \in A$, which 
yield the same pure output state $\ket{\psi}$ for both Bob and
Eve, or rather sequences of such state on $n$ channel uses.
%(This looks rather similar to the pure loss Bosonic channel case
%with transmissivity 50\%, and coherent state signals that was
%considered in~\cite{Guha-et-al:enigma}.)
%
Now, similar to the protocol in Section~\ref{sec:main}, use 
input states which result in either $Z$-basis or $X$-basis 
eigenstates output (equally for Bob and Eve, obviously), on 
small blocks of size $k$ in $n$ transmissions, so that only 
$\frac{n}{k}$ bits of key are
required for both Alice and Bob to know the basis and to have
a perfect communication channel. On the other hand, Eve, without
this bit of basis information faces uniformly random states in
one of two mutually unbiased bases in dimension $d^k$, namely
the $Z$ eigenstates $\ket{i^k}$ and the $X$ eigenstates $\ket{\varphi_{i^k}}$. 
Hence, for any POVM $Q=(Q_j)$ on $E^k$,
\[
  H(I|J) \geq \min_\psi S\bigl(\cM(\psi)\bigr) = 1 + \frac12 k \log d,
\]
and so we can invoke Proposition~\ref{prop:high-order} and 
Corollary~\ref{cor:accessible-conditional-Hmin}: For $\ell = \frac{n}{k}$ 
uses of this scheme we obtain, for the measurement outcomes $J$ of 
an arbitrary POVM $(Q_j)$,
%
%Lemma~\ref{lemma:H-alpha-S} in Appendix~\ref{app:simple}
%\[
%  H_\alpha(I|J) \geq \min_\psi S_\alpha\bigl(\cM(\psi)\bigr)
%                \geq 1 + \frac12 k \log d - 16(\alpha-1)(k\log d)^2,
%\]
%where $\cM$ is the qc-channel corresponding to the POVM
%$\left(\frac{1}{2}\proj{i^k}\right) \stackrel{.}{\cup} \left(\frac{1}{2}\proj{\varphi_{i^k}}\right)$.
%%
%Choose $\alpha = 1 + \frac{\delta}{16 k\log d}$ ($\delta < \frac12$), so that
%for every measurement on $E^k$ with outcomes $J$
%\[
%  H_\alpha(I|J) \geq 1 + \left(\frac12 - \delta\right) k \log d.
%\]
%
%Now the argument progresses as in Section~\ref{sec:main}: The minimum output
%$\alpha$-entropy is additive for $\ell = \frac{n}{k}$ uses of this scheme, and
%translating to the min-entropy as before (Lemma~\ref{lemma:H-min-H-alpha}
%in Appendix~\ref{app:simple}), we get
%
\[\begin{split}
  H_{\min}^{\epsilon}(I^\ell|J) &\geq n\left(\frac12 - \frac{\delta}{k}\right)\log d 
                                        + \frac{n}{k} 
                                        - \frac{16 k^2(\log d)^2}{\delta} \log\frac{1}{\epsilon}  \\
                                &\geq n\left(\frac12 - \frac{\delta}{k}\right)\log d,
\end{split}\]
where for the last line we have made the choice
\[
  k = \sqrt[3]{n\frac{\delta}{16 (\log d)^2 \log\frac{1}{\epsilon}}}.
\]

Now the argument progresses as in Section~\ref{sec:main}: on top of the 
$\frac{n}{k}$ bits of key, we use another $o(n)$ 
bits for the randomness extractor. The key rate goes to zero as long 
as $k \longrightarrow \infty$.
We see that we can achieve the locking rate $\frac12 \log d$, 
and even let $\frac{\delta}{k}$ and $\epsilon$ go to $0$ sufficiently slowly: 
for instance, constant $\delta$ and $\epsilon = 2^{-n^\gamma}$, 
with any $\gamma < 1$. 
Thus we have proved the following theorem.

\begin{theorem}
  \label{thm:P=0}
  The symmetric subspace channel $\cS:\cL(\text{Sym}^2(\CC^d))\longrightarrow \cL(\CC^d)$
  has zero private capacity, $P(\cS) = 0$ (since it gives a copy of every
  output state of Bob to Eve), but $L_W(\cS) \geq \frac12 \log d$.
  \qed
\end{theorem}

\medskip
This result dramatically improves the argument in~\cite{KRBM-locking}
and~\cite{ChristandlEHHOR}:
Here we have a regular quantum cryptographic system, which could be
BB84 sending a copy of each of the four states not only to Bob but
also to Eve, very much like in the ``photon number splitting 
attack''~\cite{HuttnerImotoGisinMor} -- hence there can be no 
private communication capacity
in the present setting of only a sublinear pre-shared key. But it can be
perfectly good key as judged by the accessible information of Eve. 
%
%Observe furthermore that in contrast to the examples in 
%Section~\ref{sec:main}, here the sublinear pre-shared key is definitely 
%needed. On the other hand, it is easy to see that the sublinear pre-shared
%key cannot increase $P(\cN)$ to a positive private capacity. 

%\bigskip
\begin{remark}
  \normalfont
  It may even be possible to show that for large enough $d$, the locking
  capacity $L_W(\cS)$ can be arbitrarily close to $C(\cN) = \log d$,
  by using encodings into $m > 2$ many bases. What we would
  need is that repeating the basis a small number of times would still
  result in a ``strong'' uncertainty relation (cf.~\cite{WehnerWinter:review}),
  as it was shown for $k=1$ in~\cite{HLSW:random} and~\cite{FHS:geometric-uncertainty}.
  
  To be precise, denote the bases $(U_t\ket{i})_{i=1}^d$, with unitaries
  $U_t$ ($t=0,\ldots,m$), so that we get $m$ ``repeated'' bases 
  $U_t^{\ox k}$, each of which defines
  an orthogonal measurement 
  $\bigl( M^{(t)}_{i_k} = U_t^{\ox k}\proj{i^k} {U_t^{\ox k}}^\dagger \bigr)$.
  The question then is, whether it is possible to find $U_t$ such that
  for all states $\psi$ on $(\CC^d)^{\ox k}$,
  \[
    \frac{1}{m} \sum_{t=1}^m H\bigl(\{\tr\psi M^{(t)}_{i^k}\}_{i^k}\bigr)
                                                         \geq c(k,m) \log d^k,
  \]
  with $c(k,m) \rightarrow 1$ as $k \rightarrow \infty$ and $m \leq 2^{k^c}$
  for some $0<c<1$. With such an uncertainty relation in hand, we could go 
  through the proof of Theorem~\ref{thm:P=0}, letting $n = \text{poly}(k)$ as before
  and thus using $\sim \frac{n}{k}k^c = o(n)$ bits of key, while attaining
  a locking rate of $c(k,m) \log d$.
\end{remark}

\section{Conclusion and possible further developments}
\label{sec:further}
Our results on separations between $P(\cN)$ and $L_W(\cN)$, 
Theorems~\ref{thm:main} and~\ref{thm:P=0}, and Example~\ref{example:MUB}, 
have an interpretation in terms of quantum key distribution (QKD):
$\cN$ may be the effective channel between Alice and Bob if
Eve applies a so-called \emph{collective attack}, the same and 
known isometry $V$ to all transmissions. In fact, the security
definition of the weak locking capacity is the old-style notion
put forward in the very first complete analyses of BB84 and related
protocols~\cite{Mayers:BB84}, cf.~the historical account in the nice
review~\cite{Scarani-review} (footnote 20). Indeed, in these older texts,
it was assumed that it is enough to bound the ``knowledge of Eve
about the key'', understood as the (Shannon) information she can 
obtain by making a suitable measurement after collecting all sorts
of quantum and classical systems during the protocol. The definition
of $L_W(\cN)$ is essentially based on taking this notion of cryptographic
security of the key literally.
It was only with the discovery of quantum information 
locking~\cite{I-locking} and its subsequent 
development~\cite{HLSW:random,FHS:geometric-uncertainty} that
it was eventually understood that this is a very badly behaved 
security criterion, in particular not composable, and subject
to chosen plaintext attacks (where the eavesdropper has side-information 
about the message to be transmitted)~\cite{KRBM-locking,ChristandlEHHOR}. 
The timing problem of \emph{when} the measurement should take place,
and hence whether side-information becomes available before or
after it, is at the heart of this issue, and has been investigated
in its own right~\cite{SIPI}.
The new, modern, information theoretic security definition~\cite{Renner:PhD}
is at the basis of the notion of private capacity $P(\cN)$.

Furthermore, Theorem~\ref{thm:P=0} shows that it is possible for
the locking capacity to be positive where ``evidently'', due to the
symmetry of the channel between legal and eavesdropping users, there
can be no secrecy. The coding scheme may even be interpreted in the
context of the famous \emph{photon number splitting attack} on
coherent state based QKD protocols~\cite{HuttnerImotoGisinMor}:
the protocol of Theorem~\ref{thm:P=0} is as if Alice \emph{always}
prepares a state of two photons, in fact two identical copies of her
chosen polarization -- and naturally Bob and Eve each get one.

\medskip
The main open question about the Hadamard channels considered
in Section~\ref{sec:main} is, whether $S_{\text{acc}}(I|E)$,
or in other words, the constrained minimum output entropy of
the associated POVM for a given distribution $(p_i)$ of the
inputs,
\[
  \widehat{H}(\cM|\sigma) 
        = \min_{\{q_j,\psi_j\}} \sum_j q_j H\bigl(\cM(\psi_j)\bigr)
                    \quad\text{s.t.}\quad \sum_j q_j \psi_j = \sigma,
\]
is an achievable locking rate.

The obvious first step to try would be
to consider the R\'{e}nyi entropic version of this,
\[
  \widehat{H}_\alpha(\cM|\sigma) = \min H_\alpha(I|J) 
                           \quad\text{s.t.}\quad \sum_j q_j \psi_j = \sigma,
\]
where
\[
  H_\alpha(I|J) 
  = - \frac{\alpha}{\alpha-1} 
      \log \left( \sum_j q_j \left( \sum_i (\tr \psi_j M_i)^\alpha \right)^{1/\alpha} \right)
\]
is the conditional $\alpha$-entropy (cf.~\cite[Def.~4]{Mueller-Lennert-et-al},
where the classical case is attributed to Arimoto~\cite{Arimoto}). 
Note that it relates to the
smooth conditional min-entropy and the conditional von Neumann entropy in analogous
ways as the non-conditional versions~\cite{Tomamichel:PhD}, here stated
as Lemmas~\ref{lemma:H-min-H-alpha} and \ref{lemma:H-alpha-S} in
Appendix~\ref{app:simple}.

What is missing is an additivity proof. So this is the question: for
two POVMs $M^{(1)}$ and $M^{(2)}$, and states $\sigma_1$ and $\sigma_2$,
does it hold that
\begin{equation}
  \label{eq:constrained-alpha-add}
  \widehat{H}_\alpha(\cM^{(1)}\ox\cM^{(2)}|\sigma_1\ox\sigma_2) 
%    \stackrel{\text{?}}{=} 
    = \widehat{H}_\alpha(\cM^{(1)}|\sigma_1) + \widehat{H}_\alpha(\cM^{(2)}|\sigma_2) \, ?
\end{equation}
[Note that ``$\leq$'' is trivially true.] If that is the case, we are
done, by substituting $\widehat{H}_\alpha(\cM|\sigma)$ for the
simpler, and smaller, $\widehat{H}_\alpha(\cM)$, in the proof of
Theorem~\ref{thm:main}. In the limit $\alpha\rightarrow 1$ this
is true by~\cite[Lemma~3]{King:additive}; see also the proof of our
Lemma~\ref{lemma:S-acc-add}.

This would give the weak locking capacity for those channels,
since the achievable rate, optimized over all input distributions,
\emph{i.e.}~$\displaystyle{L_W(\cN) \geq \max_{(p_i)} S_{\text{acc}}(I|E)}$,
would then match the multi-letter converse 
from~\cite{Guha-et-al:enigma} for these channels, which by
Proposition~\ref{prop:upper-bound} for the present channels 
simplifies to $\displaystyle{L_W(\cN) \leq \max_{(p_i)} S_{\text{acc}}(I|E)}$.

It should be noted, that what we really need is a lower bound
on the smooth min-entropy of the i.i.d.~$I^n$ conditioned on the
measurement outcomes $J$ from the POVM on $E^n$, of the form
\begin{equation}
  \label{eq:constrained-min-entropy}
  H_{\min}^\epsilon(I^n|J) \stackrel{\text{?}}{\gtrsim} n\bigl( S_{\text{acc}}(I|E) - \delta \bigr),
\end{equation}
analogous to~\cite{Damgaard:high-order}. This would be implied by 
eq.~(\ref{eq:constrained-alpha-add}) being true, along the lines
of the proof in Appendix~\ref{app:simple}. 
But even if the additivity fails there might be a direct proof 
of eq.~(\ref{eq:constrained-min-entropy}).

\medskip
Going on to more general channels, we could then approach the problem
of how tight is the upper bound on $L_W(\cN)$ in terms of the regularization of
\[
  L_W^{(u)}(\cN) = \max_{\{p_x,\rho_x\}} I(X:B) - I_{\text{acc}}(X:E).
\]
This seems a difficult question, as it has to be noted
that there is no obvious way how to attain it as a rate for a locking
code. The problem lies in the term $I(X:B)$, which suggests that we 
should select a code for the channel on blocks of length $n$, rather
than i.i.d.~copies $X^n$, cf.~\cite{Devetak:P+Q,CaiWinterYeung:P}. 
But when the i.i.d.~ensemble structure
of the inputs $X^n$ is disrupted, the accessible information 
$I_{\text{acc}}(X^n:E^n)$ can possibly change dramatically, because
of the very locking effect~\cite{I-locking}.

A possible way forward would be to allow the use of another, private,
channel, let's say for concreteness a noiseless channel of sufficiently large 
dimension $k$. Then, Alice can use an ensemble decomposition of the 
i.i.d.~$X^n$ into good codes for the channel $\cN^{\ox n}$ to Bob, 
choose one of them at random and inform him about the choice over 
the auxiliary channel; any $\log k \gtrsim H(X|B)$ will do. Via
the code she can send $I(X:B)$ bits per channel use to Bob. 
For Eve, on the other hand, the noiseless channel $\id_k$ does not
yield any information, and $\cN$ appears to be used with i.i.d.~$X^n$
from her point of view. Then, assuming the additivity hypothesis
(\ref{eq:constrained-alpha-add}), or rather the min-entropy uncertainty
relation (\ref{eq:constrained-min-entropy}) above, this ``raw key'' can 
be hashed down to $S_{\text{acc}}(X|E)$ locked bits per channel use.
By the same argument of ``running the extractor backwards'' as
in the proof of Theorem~\ref{thm:main}, we thus would get
$L_W(\cN\ox\id_k) \geq L_W^{(u)}(\cN) + \log k = L_W^{(u)}(\cN\ox\id_k)$.

This sketch of a proof should suffice to show the following:
\begin{theorem}
  \label{thm:activated-locking-cap}
  If the additivity hypothesis (\ref{eq:constrained-alpha-add}),
  or more specifically, the min-entropy uncertainty relation
  (\ref{eq:constrained-min-entropy}) is true, then for $\cN$
  the complement of a cq-channel,
  \[
    L_W(\cN) = \max_{(p_i)} S_{\text{acc}}(I|E) = L_W^{(u)}(\cN).
  \]
  
  Furthermore, for an arbitrary channel $\cN$, 
  the \emph{activated} (or \emph{amortized}) \emph{weak locking capacity} 
  $\overline{L}_W(\cN) := \sup_k L_W(\cN\ox\id_k) - \log k \geq L_W(\cN)$
  is given by
  \[
    \hspace{5.6cm}
    \overline{L}_W(\cN) = \sup_n \frac1n L_W^{(u)}(\cN^{\ox n}).
    \hspace{5.4cm}\qed
  \]
%  \qed
\end{theorem}

\medskip
Of course, we do not know at this point whether $L_W(\cN) = \overline{L}_W(\cN)$
for all channels. Note however that strict inequality would imply that
$L_W$ is non-additive even when combining a noisy channel with a noiseless
one. On the other hand, maybe $\overline{L}_W$ is a more natural
definition of locking capacity, since the (amortized) use of the noiseless
channel really amounts to allowing a linear secret key rate, rather than a
sublinear amount, but letting the users pay for it.

\medskip
Regarding the original locking capacity papers~\cite{Lloyd-aenigma,Guha-et-al:enigma},
a very interesting problem would be to find a nontrivial lower bound on
the weak locking capacity of Gaussian channels, such as the pure
loss Bosonic channel. Indeed, maybe the $50\%$-lossy channel, which 
has private capacity $0$, can be analyzed along the lines of 
Section~\ref{sec:P=0}? Note however, that from~\cite{Guha-et-al:enigma}
we have a \emph{constant} upper bound on its weak locking capacity,
irrespective of the input power, at least for coherent state encodings. 
It may be observed 
here that indeed~\cite[Thms.~26 and 27]{Guha-et-al:enigma} hold
also for more general encodings into statistical mixtures of coherent
states, which would be the kind of code that our main constructions
would yield, even starting from pure coherent state ensembles.

\medskip
Finally, to close this long list of open questions, let us turn to the
strong locking capacity~\cite{Guha-et-al:enigma}, which we have not 
touched upon at all in this paper. In fact, there might be link
between weak and strong locking, suggested by a simple generalization
of the symmetric channel (\ref{eq:symmetric}):
\[
  \cS_k:\cL\bigl(\text{Sym}^k(B)\bigr) \longrightarrow \cL(B),
\]
with $B \simeq \CC^d$ and $E \simeq \text{Sym}^{k-1}(B) \subset (\CC^d)^{\ox k-1}$,
which has as its Stinespring dilation the isometric embedding of
the $k$-fold symmetric subspace $A = \text{Sym}^k(B)$ 
into $B \otimes E$. The generalization of the scheme in
Section~\ref{sec:P=0} would be to encode information into
$\ket{\psi}^{\ox k} \in A$, so that Bob gets one, Eve instead
$k-1$ copies of $\ket{\psi}$, chosen from one of several
bases determined by the pre-shared key. 
It seems quite reasonable to expect that all of these 
(anti-degradable) channels have positive weak locking capacity. 
But weak locking for $\cS_k$ implies strong locking for $\cS_{k-1}$, 
and so we expect that $L_S(\cS_k) \geq L_W(\cS_{k+1}) > 0$
for all $k\geq 2$.

\section*{Acknowledgments}
I thank Mark Wilde for enlightening discussions on 
information locking, for introducing me to locking capacities and for 
first raising the problem of separating the private capacity from the 
weak locking capacity.
The keen interest he and Saikat Guha took in this project helped 
immensely to develop the ideas of the present paper.

This research was supported by the European Commission (STREP ``RAQUEL''),
the European Research Council (Advanced Grant ``IRQUAT''),
and the Spanish MINECO (project FIS2008-01236) with FEDER funds.
Part of this work was done during the programme ``Mathematical
Challenges in Quantum Information'' (MQI), 27/8-20/12/2013, at the 
Isaac Newton Institute in Cambridge, whose hospitality during
the semester is gratefully acknowledged.

\appendix

\section{High-order min-entropy uncertainty relation \protect\\ via additivity of output R\'{e}nyi entropies}
\label{app:simple}
Here we give a simple direct proof of Proposition~\ref{prop:high-order}.
In~\cite{Damgaard:high-order}, it was first shown using Azuma's
inequality for tails of martingales and a nontrivial truncation
trick. The following proof rests
on lower bounding the smooth min-entropy in terms of
R\'{e}nyi entropies, and lower bounding the latter in terms of 
von Neumann entropies. This idea can be traced back 
to~\cite{ColbeckRennerTomamichel:AEP}. The relevant lemmas are 
stated here for completeness, and they are direct corollaries of 
the citations given.

\begin{lemma}[{Renner/Wolf~\cite{RennerWolf}}]
  \label{lemma:H-min-H-alpha}
  For any state $\rho$ and $\alpha > 1$,
  \[
    \hspace{5.3cm}
    H_{\min}^{\epsilon}(\rho) \geq H_\alpha(\rho) - \frac{1}{\alpha-1}\log\frac{1}{\epsilon}.
    \hspace{5.2cm}
    \qed
  \]
%  \qed
\end{lemma}

\begin{remark}
  \normalfont
  Under smoothing with respect the purified distance, the above relation
  would read 
  \[
    H_{\min}^{\epsilon}(\rho) \geq H_\alpha(\rho) - \frac{1}{\alpha-1}\log\frac{2}{\epsilon^2}.
  \]
  (Cf.~Tomamichel~\cite[Prop.~6.2]{Tomamichel:PhD}.)
  As pointed out already, in the present paper we are using instead the smoothing 
  w.r.t.~the trace norm.
\end{remark}

\begin{lemma}[{Tomamichel~\cite[Lemma 6.3]{Tomamichel:PhD}}]
  \label{lemma:H-alpha-S}
  For any state $\rho$ on a $d$-dimensional Hilbert space,
  and $1 < \alpha < 1 + \frac{\log 3}{4\log \nu}$,
  with $\nu = 2+\sqrt{d}$,
  \[
    H_\alpha(\rho) \geq H(\rho) - 4(\alpha-1)(\log \nu)^2.
  \]
  A simplified version reads thus: For $1 < \alpha < 1 + \frac{\log 3}{16\log d}$
  and $d \geq 2$,
  \[
    \hspace{5cm}
    H_\alpha(\rho) \geq H(\rho) - 16(\alpha-1)(\log d)^2.
    \hspace{4.8cm}
    \qed
  \]
%  \qed
\end{lemma}

\medskip
\noindent
%\begin{proof}
%{\!\!\!\!\textbf{(of Proposition~\ref{prop:high-order})}\ }
{\bf Proof (of Proposition~\ref{prop:high-order})}\ 
By Lemma~\ref{lemma:H-min-H-alpha}, for arbitrary $n$ and state $\psi$,
\[
  H_{\min}^\epsilon\bigl( \cM^{\ox n}(\psi) \bigr)
       \geq H_\alpha\bigl( \cM^{\ox n}(\psi) \bigr) - \frac{1}{\alpha-1}\log\frac{1}{\epsilon}.
\]
On the other hand, by the additivity of the minimum output $\alpha$-R\'{e}nyi
entropy of qc-channels, and more generally entanglement-breaking 
channels~\cite{King:additive-0,King:additive},
\[
  H_\alpha\bigl( \cM^{\ox n}(\psi) \bigr) \geq n \widehat{H}_\alpha(\cM),
\]
with 
$\displaystyle{\widehat{H}_\alpha(\cM) 
   := \min_{\psi\text{ state}} H_\alpha\bigl(M(\psi)\bigr)}$.
Hence, using now Lemma~\ref{lemma:H-alpha-S},
\[
  H_{\min}^\epsilon\bigl( \cM^{\ox n}(\psi) \bigr)
       \geq n \bigl( \widehat{H}(\cM) - 16(\alpha-1)(\log d)^2 \bigr) 
              - \frac{1}{\alpha-1}\log\frac{1}{\epsilon},
\]
as along as $\alpha$ is close enough to $1$.
%$1 < \alpha < 1 + \frac{\log 3}{16\log d}$.
%
Letting $\alpha = 1+ \frac{\delta}{16(\log d)^2}$, we conclude
\[
  \hspace{3.5cm}
  H_{\min}^\epsilon\bigl( \cM^{\ox n}(\psi) \bigr)
       \geq n \bigl( \widehat{H}(\cM) - \delta \bigr) 
              - {16(\log d)^2} \frac{1}{\delta} \log\frac{1}{\epsilon}.
  \hspace{3.3cm}
  \qed
\]
%and the proof is finished.
%\end{proof}

%\bibliographystyle{plain}
%\bibliography{Ref}

\end{document}